\documentclass[sigplan,9pt]{acmart} 

\setcopyright{none}
\acmConference{Formal Techniques for Java-like Programs}{2018}{(FTfJP)}
\acmYear{2018}
\acmISBN{} 
\acmDOI{} 
\startPage{1}

\usepackage{booktabs}
\usepackage{dsfont}
\usepackage{amsmath,amssymb,amsfonts,mathtools}
\usepackage{pifont}
\usepackage{stmaryrd}
\usepackage{latexsym}
\usepackage{paralist}
\usepackage{wrapfig}
\usepackage{subcaption}
\usepackage{caption} 
\usepackage{adjustbox}
\usepackage{color}
\usepackage{listings}
\usepackage{lstcustom}

\newcommand{\mymathlcl}[1]{
$\begin{array}{lcl}
#1
\end{array}$
}

\newcommand{\keyword}[1]{\textsf{\small{#1}}}

\newcommand{\grmeq}{::=}
\newcommand{\grmor}{\mid}

\newcommand{\getInPos}[1]{[#1]}
\newcommand{\ands}{\ \land\ }

\newcommand{\defk}{\keyword{def}}

\newcommand{\nullk}{\keyword{null}}

\newcommand{\classk}{\keyword{class}}

\newcommand{\layoutk}{\keyword{layout}}
\newcommand{\recordk}{\keyword{rec}}

\newcommand{\newk}{\keyword{new}}

\newcommand{\poolsk}{\keyword{pools}}
\newcommand{\localsk}{\keyword{locals}}

\newcommand{\thisk}{\keyword{this}}

\newcommand{\nonek}{{\keyword{none}}}

\newcommand{\dom}[1]{\operatorname{dom}(#1)}

\newcommand{\many}[1]{#1s}

\newcommand{\concretetype}{\mathit{t}}
\newcommand{\abstracttype}{\mathit{pt}}

\newcommand{\pooltype}[1]{[#1]}
\newcommand{\ownerT}[2]{ #1\langle #2 \rangle }

\newcommand{\scallc}[3]{#1.#2( #3 ) }
\newcommand{\freadc}[2]{ #1 . #2 }
\newcommand{\fwritec}[3]{ #1 . #2 = #3 }
\newcommand{\newc}[1]{ \newk\ #1 }

\newcommand{\setDef}[1]{\mathit{ #1 }}
\newcommand{\cds}{\setDef{cd}}
\newcommand{\tds}{\setDef{ld}}
\newcommand{\cls}{\setDef{rd}}
\newcommand{\fds}{\setDef{fd}}
\newcommand{\pds}{\setDef{pd}}
\newcommand{\pbds}{\setDef{pbd}}
\newcommand{\mds}{\setDef{md}}
\newcommand{\lds}{\setDef{vd}}
\newcommand{\lpds}{\setDef{lpd}}
\newcommand{\lvds}{\setDef{lvd}}

\newcommand{\heap}{\mathcal{X}}
\newcommand{\pool}{\pi}

\newcommand{\sframe}{\Phi}
\newcommand{\object}{\omega}
\newcommand{\record}{\rho}
\newcommand{\cluster}{\kappa}
\renewcommand{\address}{\alpha}

\newcommand{\hlv}{\beta} 
\newcommand{\natnumber}{\mathds{N}}

\newcommand{\funk}[1]{\keyword{fun}}

\newcommand{\rulename}[1]{\textsc{\footnotesize{#1}}}

\newcommand{\hlrule}[3]{
  \begin{array}{l}
    \rulename{[#1]}\hfill
    \\ \begin{array}{c}#2 \\\hline #3\end{array}
  \end{array}}

\newcommand{\typeRule}[3]{
	\begin{array}{l}
		\rulename{[#1]}\hfill
		\\ \begin{array}{c}#2 \\ \hline #3\end{array}
	\end{array}}

\newcommand{\reduces}[2]{#1\leadsto#2}

\newcommand{\types}[3]{#1\vdash#2\colon #3} %

\newcommand{\fields}[1]{\mathcal F\!\mathit{s}({#1})}
\newcommand{\flookup}[2]{\mathcal F({#1}, {#2})}
\newcommand{\plookup}[2]{\mathcal P({#1}, {#2})}
\newcommand{\mlookup}[2]{\mathcal M({#1}, {#2})}
\newcommand{\getclass}[1]{\mathcal C(#1)}
\newcommand{\getlayout}[1]{\mathcal L(#1)}
\newcommand{\offset}[2]{\mathcal W(#1, #2)}

\newcommand{\owners}[1]{\mathcal{P}\!\mathit{s}(#1)}

\newcommand{\wftype}[2]{#1\vdash#2}

\newcommand{\wfhlclass}[2]{#1\vdash #2}
\newcommand{\wflayout}[2]{#1 \vdash #2}

\newcommand{\wfHL}[3]{#1 \vDash #2, #3}
\newcommand{\wfHeap}[1]{\vDash #1}

\newcommand{\wfagrees}[3]{#1 \vDash #2\triangleleft  #3}
\newcommand{\wfweakOK}[3]{#1 \vDash #2\colon  #3}

\newcommand{\context}{\Gamma}

\lstdefinelanguage{ohmm}{
  keywords=[1]{class,def,if,then,else,new,for,in,do,let,null,true,false,void,boolean,number,string,while,rec,of,none,None,this,not,gt,layout,pool,heap,pools,locals,return},
  breaklines=false,
  morecomment=[l]{//},
  morecomment=[s]{/*}{*/},
  morestring=[b]",
  tabsize=2,
  literate=
    {->}{$\rightarrow$}2
    {::}{$\parallel$}2
    {<<}{$\langle$}1
    {>>}{$\rangle$}1
    {leftbrace}{$\{$}1  
    {rightbrace}{$\}$}1
    {_0}{$_\textsf{\scriptsize 0}$}1
    {_1}{$_\textsf{\scriptsize 1}$}1
    {_2}{$_\textsf{\scriptsize 2}$}1
    {_3}{$_\textsf{\scriptsize 3}$}1
    {_4}{$_\textsf{\scriptsize 4}$}1
    {_5}{$_\textsf{\scriptsize 5}$}1
}
\definecolor{bluefield}{rgb}{0.32,0.65,1}
\definecolor{greenfield}{rgb}{0.44,0.75,0.25}
\definecolor{yellowfield}{rgb}{0.86,0.74,0.14}
\definecolor{orangefield}{rgb}{0.87,0.42,0.06}

\newcommand{\eg}{\emph{e.g.}}
\newcommand{\ie}{\emph{i.e.}}
\renewcommand{\c}[1]{\lstinline@#1@}

\renewcommand*{\proofname}{Proof sketch}

\begin{document}

\title{Safely Abstracting Memory Layouts}

\author{Juliana Franco}
\email{juliana.franco@microsoft.com}
\affiliation{\institution{Microsoft Research Cambridge}}

\author{Alexandros Tasos}
\email{a.tasos17@imperial.ac.uk}
\affiliation{\institution{Imperial College London}}

\author{Sophia Drossopoulou}
\email{s.drossopoulou@imperial.ac.uk}
\affiliation{\institution{Imperial College London}}

\author{Tobias Wrigstad}
\email{tobias.wrigstad@it.uu.se}
\affiliation{\institution{Uppsala University}}

\author{Susan Eisenbach}
\email{s.eisenbach@imperial.ac.uk}
\affiliation{\institution{Imperial College London}}

\renewcommand{\shortauthors}{J. Franco, A. Tasos, S. Drossopoulou,
  T. Wrigstad, S. Eisenbach}

\newcommand{\SHAPES}{{\textsf{SHAPES}}}
\newcommand{\Pad}{\vspace*{1ex plus .5ex minus .5ex}}
\newcommand{\Tigthen}{\vspace*{-1ex plus .5ex minus .5ex}}

\begin{abstract}
  Modern architectures require applications to make effective use
  of caches to achieve high performance and hide memory latency.
  This in turn requires careful consideration of placement of data
  in memory to exploit spatial locality, leverage hardware
  prefetching and conserve memory bandwidth.
  In unmanaged languages like C++, memory optimisations are
  common, but at the cost of losing object abstraction and memory
  safety.
  In managed languages like Java and C\#, the abstract view of
  memory and proliferation of moving compacting garbage collection
  does not provide enough control over placement and layout.

  We have proposed \SHAPES{}, a type-driven abstract placement
  specification that can be integrated with object-oriented languages
  to enable memory optimisations. \SHAPES{} preserves both memory and
  object abstraction. In this paper, we formally specify the \SHAPES{}
  semantics and describe its memory safety model.
\end{abstract}

\maketitle

\section{Introduction}
Modern computers use hierarchies of memory caches to hide memory
latency~\cite{wulf1995hitting}. Accessing data from the
bottom of the hierarchy --- from main memory --- is often an order of magnitude slower
than reading from the top --- from the fastest memory cache. Whenever the CPU needs to read
from a particular memory location, it first tries to read from the
top-most cache. If it succeeds, then the data is delivered at very
little cost. Otherwise, a \emph{cache miss} has occurred, and the CPU
tries to obtain the data from the next cache level. In the worst 
case scenario, data must be fetched from main memory. This commonly causes
the data --- and 
adjacent data --- to be cached. How much adjacent data is cached and how
much data can be cached at what level varies across different hardware.
Caches typically have sizes in kilobytes or a few megabytes, meaning that
bringing only relevant data into cache is important for proper cache utilisation.

In programs whose performance is memory-bound,
reorganising data in memory to 
reduce the number of cache misses can have
enormous performance impact \cite{wulf1995hitting, calder-1998}. Writing cache-friendly
code typically amounts to allocating contiguously
data that is accessed together and in patterns recognisable to the \emph{hardware prefetcher},
allowing it to anticipate a program's data needs ahead-of-time. 
As a concrete example, in a language like C, one
can allocate an array of structures (not pointers to structures), and
keep all the objects of the same data structure in that array, in the
order in which the objects are most frequently accessed. Thus, when an object is
fetched to cache, the next one to be read, is potentially fetched as
well. To improve cache utilisation, programmers often
\emph{split objects} across multiple arrays, in order to keep
the \emph{hot fields} --- those most often used --- of different
objects together. This is called transforming an array of structures
into a structure of arrays. Both layouts are depicted in
Figure~\ref{fig:layouts}. As a consequence of this representational
change, a complex datum is no longer possible to reference by a single
pointer. 

\begin{figure}[ht]
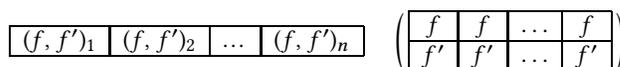

\begin{tabular}{|c|c|c|c|}
  \hline 
  ($f$, $f'$)$_1$ & ($f$, $f'$)$_2$ & \ldots & ($f$, $f'$)$_n$ \\
  \hline 
\end{tabular}
\quad
$\bigg(\,
\begin{array}{|c|c|c|c|}
  \hline 
  f & f & \ldots & f \\
  \hline 
  f' & f' & \ldots & f' \\
  \hline 
\end{array}
\,\bigg)$
\caption{Array of structs (left) vs. struct of arrays (right).
  $(f, f')$ denotes a struct (object) with fields $f$ and $f'$. Left,
  each \emph{cell} is an object. Right, each \emph{column} holds the data for one object.}
\label{fig:layouts}
\end{figure}

Optimisations like the one sketched above are common in
low-level languages such as C or C++, where programmers have control
over where and how memory is allocated. Unfortunately,
manipulating data layouts in memory is complex, error-prone, and
pollutes the logic of the program with layout concerns.

In managed languages like Java and C\#, the abstract level
at which memory is handled, as well as the presence of moving garbage
collectors, make many memory optimisations impossible. 
To this end, we have proposed a
type-driven approach to abstractly express layout concerns, called \SHAPES{} \cite{Onward2017}, whose aim is to allow the
application of layout optimisations in 
high-level, managed, object-oriented languages. \SHAPES{} protects the object abstraction,
allowing an object to be referenced by --- and manipulated via -- a
single pointer, regardless of its physical representation.

\paragraph{Contributions}

In previous work, we have introduced the \SHAPES{} idea, and recap
the key moving parts in \S\,2. In this paper, we make the
following contributions:
\Pad
\begin{compactitem}[--]
\item We formalise \SHAPES{} (\S\,\ref{sec:formalism}) through a high-level object-oriented
  calculus where classes are parameterised by layout
  specifications. The formalism is an important component of
  proving soundness of \emph{uniform access} --- allowing programmers to
  write \c{x.f} to access the field \c{f} of the object $o$
  pointed to by \c{x}, regardless of how $o$ is laid out in
  memory.
  We specify the dynamic semantics (\S\,\ref{sec:dynamics}) as
  well as the static semantics (\S\,\ref{sec:statics}).

\item We define the main invariants of \SHAPES{} in terms of type
  safety and memory safety (\S\,\ref{sec:meta}).
\end{compactitem}
\Pad
\noindent
\S\,\ref{sec:related} discusses related work and 
\S\,\ref{sec:final} concludes.


\section{Getting into \SHAPES{}}

The \SHAPES{} vision gives the programmer control of how data
structures are laid out in memory in  object-oriented,
managed programming languages. Although programmers have
control over how their data structures are allocated, memory is still 
abstract, and all operations are type and memory safe.

The \SHAPES{} approach
makes classes parametric with \emph{layouts} --- abstract regions (\emph{pools})
which collect their objects together in physical memory, optionally
using some form of \emph{splitting strategy} to keep hot fields together in memory. The design
delays the choice of the physical representation of a data structure to
instantiation-time, rather than declaration-time.
Thus, it is possible to reuse the same data structure
declaration with multiple layouts. Reflecting layout choices in
types thus serves the dual purpose of separating layout concerns
from business logic and guiding compilers' generation of efficient
code for allocating and accessing data. It is also key to a
uniform access model (\eg{} \c{x.f}), regardless of how an object
may be laid out in memory.
In a typical statically typed OO language, \c{x.f} is translated
by a compiler into loading a value at the address pointed to by
\c{x} plus the offset of \c{f} in the type of \c{x}. Because the
same \c{x.f} may manipulate \c{x}'s that point to objects with different layout, 
we must take care to propagate the layout information to ensure the
correct code is emitted, to ensure memory safety. 

\subsection{Using \SHAPES{} for Improved Cache Utilisation}

As a concrete example, Listing 1 is a partial program with a list
of \lstinline!Student!s, with 1\,000's of elements at run-time.
(For simplicity, each student holds a pointer to the next
student.)

\newcommand{\DIFF}[1]{\textcolor{magenta}{#1}}
\noindent
\begin{minipage}{.37\linewidth}
\begin{lstlisting}[numbers=none, captionpos=b, caption={A linked list.},frame=tb, aboveskip=3mm, belowskip=3mm]
class Student {
  age: int
  supervisor: Professor
  next: Student
  // ... 
  def getAge(): int {
    return age
  }
}


var list = new Student
\end{lstlisting}
\end{minipage}
\hfil
\begin{minipage}{.6\linewidth}
\begin{lstlisting}[numbers=right, captionpos=b, caption={Listing 1 using SHAPES.},frame=tb, aboveskip=3mm, belowskip=3mm]
class Student<p> {
  age: int
  supervisor: Professor
  next: Student<p>
  // ... 
  def getAge(): int {
    return age
  }
}
layout L : [Student] = rec { age, next }, *
pool P : L
var list = new Student<P>
\end{lstlisting}
\end{minipage}

Iterating over a list of students to calculate the average age in
a language like Java involves dereferencing many
\c{Student}-pointers to get the \c{age} and \c{next} fields.
However, because of how data is loaded into the cache (outlined in
\S\,1), in addition to these relevant fields, the \c{supervisor}
field and all adjacent data mapped to the same cache line will
be loaded as well. If that data is not another student in the
list, it is irrelevant to performing this calculation. 

For improved cache utilisation, we would like to:
\Pad
\begin{compactenum}[\quad a)]
\item only load the \c{age} and \c{next} fields into the cache; and
\item store \c{Student} objects from the same list adjacent in
  memory, with no interleaving from unrelated objects.
\end{compactenum}
\Pad
Listing 2 shows Listing 1 
using \SHAPES{}. We accomplish b) by introducing a \emph{pool}---a
contiguous region of storage---for holding \c{Student} objects and
use this pool to hold all the objects of our linked
list (and nothing else). This effectively stores our students like the
left of Figure~\ref{fig:layouts}. Listing 2 shows how the \c{Student}
class is parametrised by the pool parameter (Lines 1 \& 4). The pool
is created on Line 10 and connected to the list on Line 11.

To accomplish a), we alter the layout of objects in the pool
to use a structure of array storage. The layout is declared on Line 9
and used in the pool declaration on Line 10.
Our students are now stored
like the right of Figure~\ref{fig:layouts}. Note how the layout is
orthogonal from the \c{Student} class. 

Finally, if the order of the students in the pool (mostly) matches the
order of the list, iterating over the list will result in (mostly) regular
load patterns with even strides that will be detected by a
prefetcher to bring data into memory ahead-of-time. Unless this
order-alignment happens by construction, a pool-aware moving compacting
garbage collector can be used to create it. Such a collector will compact
respecting pool boundaries, and it is possible to influence moving
semantics on a per-pool basis.

%

\subsection{\SHAPES{} in a Nutshell}

\SHAPES{} adds pools, layout declarations and
parameterises classes and types with pools to
a Java-like language. The number of pools
is not fixed and pools are created at run-time. Objects may
be allocated in pools or in a ``traditional heap.''
Layout specifications specify how objects of a certain class will
be laid out in a pool of a specific layout by grouping fields
together. Allocation takes place in a pool using the layout
specification of that pool. As we do not yet model deallocation or garbage
collection, pools can be thought of as growing monotonically in this paper. 

The first parameter in a class declaration indicates the pool the object will be
allocated into.  The remaining parameters can be used in the type declarations
of fields, parameter types and return types of methods and as local variables
inside them. Similarly, types are annotated with pool arguments. In that
respect, \SHAPES{} is similar to parametric polymorphism in C++
\cite{cpp11standard} and Java \cite{java8spec} and to Ownership Types
\cite{OwnershipSurvey}.

The type system for shapes follows our desire for pools to be
\emph{homogeneous}. A pool $P$ of objects of a class $C$ is
homogeneous iff for all objects $o_1, o_2\in P$ and for all non-primitive
fields $f\in C$, $o_1.f$ and $o_2.f$ always point to objects in the same pool (or null).

To prevent the occurrence of heterogeneous pools, we require
that the types of objects in the same pool share the same set of parameters. This
falls out of the (upcoming) rules for well-formed types.
We favour homogeneous pools for the following reasons:
\Pad
\begin{compactitem}[--]
\item \emph{Smaller memory footprint.} Pointers to pools consist of a
  pool address and an index. By using homogeneous pools, the pool
  address becomes redundant, as all objects of a pool will point
  to objects allocated inside a specific pool for a specific field
  $f$.

\item \emph{Better cache utilisation}. A direct result of the
  smaller footprint.

\item \emph{No need for run-time support.} Heterogeneous pools
  allow objects allocated inside them to point to objects
  allocated in different pools. These pools may have different
  layouts, forcing a run-time check on each field access at
  run-time to obtain the correct address to load. This issue is
  eliminated for homogeneous pools.
\end{compactitem}
\Pad

\paragraph{\SHAPES{} and Performance}

In a nutshell, the \SHAPES{} design currently targets programs that
perform iteration over subsets of pooled objects, roughly in the order of access (pool allocation creating prefetcher-friendly access patterns).
If objects in data structures are spread over multiple pools, programmers must manually align the objects in the pools to be cache-friendly (although we have ideas on how to automate this in future work). When objects are split into clusters, the cost of one-off accesses to objects increase because objects are spread over multiple locations that must be loaded separately. However,
iterations become more efficient by loading only relevant data into cache. Note that performance means both execution time and power-efficiency stemming from improved cache-utilisation.





\section{Formalising \SHAPES{}}
\label{sec:formalism}

The syntax of \SHAPES{} is shown in Figure~\ref{fig:syntax}. To
simplify the formalism, we do not support programmer conveniences
including ones used in Listing 2: the \c{*} notation used to group
remaining fields (Line 9) and defaulting to store professors
in the $\nonek$ pool (Line 3).

\newcommand{\SForm}[3]{\ensuremath{\mathit{#1}\in\mathit{#2}} & ::= & \ensuremath{#3}}
\begin{figure}[b]
  $
  \begin{array}{rcl}
    \SForm{prog}{Program}{\cds^\star\ \tds^\star} \\
    \SForm{\cds}{ClassDecl}{\classk\ \ownerT{C}{\pds^+}\ \ \{\ \fds^\star \mds^\star\ \}} \\
    \SForm{\fds}{FieldDecl}{f\colon\concretetype ~ ;} \\
    \SForm{\mds}{MethodDecl}{\defk~m(x\colon\concretetype)\colon\concretetype~\{\lds ~ ; ~ e\}} \\
    \SForm{\lds}{PoolsDecl}{\poolsk ~ \lpds^\star ~\localsk ~ \lvds^\star} \\
    \SForm{e}{Expression}{\nullk \grmor x \grmor \thisk \grmor \newc \concretetype \grmor \scallc x m x } \\
& \grmor & \ensuremath{\freadc x f \grmor \fwritec x f x \grmor x = e} \\
    \SForm{y}{\setDef{PoolVariable}}{\nonek \grmor v} \\
    \SForm{x}{\setDef{LocalVariable}}{v} \\
    \SForm{\tds}{\setDef{LayoutDecl}}{\layoutk\ L\colon\pooltype C = \cls^+} \\
    \SForm{\cls}{\setDef{RecordDecl}}{\recordk\ \{ f^+\};} \\
    \SForm{\pds}{\setDef{PoolParameterDecl}}{y\colon \pbds ~ \mathit{where} ~ y \neq \nonek} \\
    \SForm{\concretetype}{\setDef{ClassType}}{\ownerT C {y^+}} \\
    \SForm{\abstracttype}{\setDef{PoolType}}{\ownerT L {y^+}} \\
    \SForm{\pbds}{\setDef{PoolBoundType}}{[\ownerT {C} {y^+}]} \\
    \SForm{\lpds}{\setDef{LocalPoolDecl}}{y\colon \abstracttype ~ \mathit{where} ~ y \neq \nonek} \\
    \SForm{\lvds}{\setDef{LocalVariableDecl}}{ v\colon \concretetype} \\
  \end{array}
  $

\caption{Syntax of \SHAPES{} where $v\in\setDef{VariableId}$, $C\in\setDef{ClassId}$, $f\in\setDef{FieldId}$, $m\in\setDef{MethodId}$, and $L\in\setDef{LayoutId}$.}
\label{fig:syntax}
\end{figure}

\paragraph{Notation and Implicit requirements.}

We append $s$ to names to indicate sequences; for instance, $\many x$
is a sequence of $x$-s, while $\many{\many x}$ is a sequence of sequences of $x$-s.

We assume that layout and class identifiers are unique within the same program,
and field and method identifiers are unique within the same class.

\paragraph{Lookup Functions}

\SHAPES{} rely on the following set of lookup functions. Their exact
definitions are in Appendix~\ref{section:lookup}.

\begin{center}
\begin{tabular}[c]{|c|p{68mm}|}
  \hline 
  \sf Fun. & \sf Used to Lookup \\
  \hline 
  $\mathcal C$ & The class declaration for a given \emph{class} identifier. \\
  $\mathcal P\!\mathit{s}$ & All the class \emph{pool parameters} of a given class. \\
  $\mathcal P$ & The \emph{bound} of a given parameter in a given class. \\
  $\mathcal M$ & The return type, the parameter type, local variables, and expression of a given \emph{method}. \\
  $\mathcal F$ & The type of a given \emph{field} in a given class. \\
  $\mathcal F\!s$ & \emph{All the field identifiers} declared in a given class. \\
  $\mathcal L$ & The layout declaration of a given layout identifier. \\
  $\mathcal W$ & The offset(s) (\ie{}, \emph{where} it is located) of a given field, within an object, or within a pool. \\
  \hline 
\end{tabular}
\end{center}








%

\subsection{Dynamic Semantics}
\label{sec:dynamics}

\SHAPES{}'s run-time entities  are defined in Figure~\ref{fig:entities}.

\begin{figure}[!t]
$
\begin{array}{rcl}
\heap \in \setDef{Heap} & = & \setDef{Address} \rightharpoonup (\setDef{Object} \cup \setDef{Pool}) \\
\sframe \in \setDef{SFrame} & = &
    \setDef{VariableId} \rightharpoonup (\setDef{Value} \cup \setDef{PoolAddress}) \\ &
    & \cup~ \{ \nonek \} \rightarrow \{ \nonek \} \\
\setDef{Object} & = &
  \setDef{ClassId} \times \setDef{PoolAddress}^+ \times \setDef{Record}\\
\setDef{Pool} & = &
    \setDef{LayoutId} \times \setDef{PoolAddress}^+ \times \setDef{Cluster}^+ \\
\record \in \setDef{Record} & = &
  \setDef{Value}^+ \\
\cluster \in \setDef{Cluster} & = &
  \setDef{Record}^+ \\
\hlv \in \setDef{Value} & = & \setDef{ObjectAddress}
  \cup \setDef{Location}
  \cup \{\nullk\} \\
\setDef{Location} & = &
  \setDef{HeapPoolAddress} \times \setDef{Index} \\
n \in \setDef{Index} & = &
  \natnumber \\
\address\in\setDef{Address} & = &
  \setDef{ObjectAddress} \uplus \setDef{HeapPoolAddress} \\
\object \in \setDef{ObjectAddress} & \\
\pool \in \setDef{PoolAddress} & = & \setDef{HeapPoolAddress} \cup \{ \nonek \} \\
\end{array}
$
\caption{Dynamic Entities of \SHAPES{}.}
\label{fig:entities}
\end{figure}
A heap ($\heap$) maps addresses to 
objects ($\object$) and pools ($\pool$). Objects consist of a class identifier (determining its type),
a sequence of pool addresses and a record ($\record$). A record is a sequence of values representing
the values in the fields of an object.

Pools consist of a layout identifier, a sequence of pool addresses and a
sequence of clusters ($\cluster$). The layout identifier determines how the
objects inside the pool are laid out. Pools can only store instances of the
class indicated by  the layout identifier. Furthermore, the layout type
determines the type of the pool and the type of the objects stored inside it.
\SHAPES{} also defines a global pool called $\nonek$ that permits objects of any
type and no splitting is performed. This is similar to the default heap
representation of \eg{} the JVM. 

The pool addresses of both objects and pools are ghost state intended to be used
for proofs; we will demonstrate in future work that they are superfluous and
serve no purpose at run-time.

The values ($\hlv$) corresponding to the fields of objects that are allocated in
pools are stored in clusters. A layout declaration designates splits; each split
contains a subset of the class' fields. Each cluster, therefore, contains the
values that correspond to the respective split's fields for all objects
allocated in that pool. Addresses of objects inside pools $(\pool, n)$ require
a pool address $\pool$ and an index $n$. The index indicates which record in
each cluster contains the values that the fields of the instance correspond to.

A frame ($\sframe$) maps variables to values. \SHAPES{} designates three kinds of
local variables:

\Pad
\begin{compactdesc}
  \setlength{\itemindent}{-6mm}
  \setlength{\labelwidth}{0mm}
\item[Local object variables] The $\thisk$ variable, and function
  parameters. These behave exactly like local variables in
  object-oriented languages.

\item[Local pool variables] These correspond to the pools that are
  constructed upon entering a method body and are initially empty.
  The reason local pool variables are defined altogether is
  because we allow reference cycles between objects that are
  allocated inside different pools.

\item[Class pool parameters] The class' pool parameters can be
  used inside method bodies as local variables for type
  declarations, object instantiations, etc.
\end{compactdesc}
\Pad

For convenience in our definitions, we define $\sframe(\nonek) =
\nonek$.

\SHAPES{} semantics rules are of the form $\reduces{\heap,\sframe, e}
{\heap',\sframe', \hlv}$. They take a heap, a stack frame and an expression and
reduce to a new heap, a new stack frame, and a new value, in a big-step manner.

\begin{figure}[!t]
\centering
\begin{gather*}
\hlrule{Value}{
}{
  \reduces{\heap, \sframe, \nullk}{\heap, \sframe, \nullk}
}
~
\hlrule{Variable}{
}{
  \reduces{\heap, \sframe, x}{\heap, \sframe,\sframe(x)}
}
\\
\hlrule{Assignment}{
  \reduces{\heap, \sframe, e}{\heap', \sframe', \hlv}
}{
  \reduces{\heap, \sframe, x = e}{\heap', \sframe'[x\mapsto \hlv], \hlv}
}
\\
\hlrule{New Object}{
  \sframe(\many y) = \nonek\cdot\many\pool
  \quad
  \object\notin\heap
  \quad
  |\fields{C}| = n
}{
  \reduces
    {\heap, \sframe, \newc{\ownerT C {\many y}}}
    {\heap[\object\mapsto(C, \nonek\cdot\many\pool, \nullk^n)], \sframe, \object}
}
\\
\hlrule{Object Read}{
  \sframe(x) = \object
  \quad
  \heap(\object) = (C, \_, \record)
  \quad
  \offset C f = i
}{
  \reduces{\heap, \sframe, \freadc x f}{\heap, \sframe, \record[i]}
}
\\
\hlrule{Object Write}{
  \sframe(x) = \object
  \quad
  \sframe(x') = \hlv
  \quad
  \heap(\object) = (C, \many\pool, \record)
  \quad
  \offset C f = i
}{
  \reduces{\heap, \sframe, \fwritec x f {x'}}
          {\heap[\object\mapsto(C, \many\pool, \record[i\mapsto\hlv])], \sframe, \hlv}
}
\end{gather*}
\caption{Operational semantics for pool-agnostic operations.}
\label{fig:semantics_no_pool}
\end{figure}

\paragraph{Operations on Pool-Agnostic Expressions}

The operational
semantics for these rules are given in Figure~\ref{fig:semantics_no_pool}. They
are unsurprising, with the exception of
\rulename{New Object} and \rulename{Method Call}. \rulename{New Object} deal with creation of unpooled objects (\ie, stored in the $\nonek$ pool). It stores
pool parameters in the objects' run-time types. 
As pool parameters are initialised in methods, and the implicit passing of the current object's pool parameters, \rulename{Method Call} is not pool-agnostic.

\paragraph{Pool-dependent operations}

The operational
semantic is given in
Figure~\ref{fig:semantics_pool}.
\rulename{New Pooled Object} allocates objects inside an existing
pool $\pool$, by appending a new record of values (all initialised
to $\nullk$) for each cluster in the pool. The notation $|\mathit{fs}_i|$
denotes the number of fields in cluster $i$ and $\nullk^{|\mathit{fs}_i|}$ a sequence of $\nullk$ values. 

By \rulename{Pooled Object Read}, accessing a field $f$ of an object (at location $n$) in a pool ($\pool$) requires
looking up the $j$:th value of the $i$:th cluster of the $n$:th object, where $i$ is the the cluster containing $f$, and $j$ the offset into that cluster according to the layout $L$ (by way of helper function $\mathcal{W}$). 
For brevity, we conflate the latter into a single 3-ary lookup: $\many\cluster[i,n,j]$.

As shown by \rulename{Pooled Object Write}, modifying a field is isomorphic.
We use a shorthand for updating, $\many\cluster[i,n,j\mapsto\beta]$, similar to lookup.

Note that the syntax for constructing objects and accessing/mutating their members is
the same regardless of layout and whether the object is allocated
in a pool or not.

\begin{figure}[!t]
\centering
\begin{gather*}
\hlrule{New Pooled Object}{
  \sframe(\many y) = \pool\cdot\_
  \quad
  \pool \neq \nonek
  \\
  \heap(\pool) = (L, \many\pool, \cluster_0~..~\cluster_n)
  \quad
  \getlayout L = (C, \many {f\!}_0~..~\many {f\!}_n)
  \\
  \heap' = \heap[\pool\mapsto(L, \many\pool,
              \cluster_0 \cdot \nullk^{|\many {f\!}_0|}
              ~..~
              \cluster_n \cdot \nullk^{|\many {f\!}_n|})]
}{
  \reduces
    {\heap, \sframe, \newc{\ownerT C {\many y}}}
    {\heap', \sframe, (\pool, |\cluster_0|)}
}
\\
\hlrule{Pooled Object Read}{
  \sframe(x) = (\pool, n)
  \quad
  \heap(\pool) = (L, \many\pool, \many\cluster)
  \quad
  \offset L f = (i, j)
}{
  \reduces{\heap, \sframe, \freadc{x}{f}}
          {\heap, \sframe, (\pool, \many\cluster[i, n, j])}
}
\\
\hlrule{Pooled Object Write}{
  \sframe(x) = (\pool, n)
  \quad
  \heap(\pool) = (L, \many\pool, \many\cluster)
  \quad
  \offset L f = (i, j)
  \\
  \sframe(x') = \hlv
  \quad
  \heap' = \heap[\pool\mapsto(L, {\many\pool, \many\cluster}[i,n,j\mapsto\hlv])]
}{
  \reduces{\heap, \sframe, \fwritec x f {x'}}{\heap', \sframe, \hlv}
}
\end{gather*}
\caption{Operational semantics for pool dependent operations.}
\label{fig:semantics_pool}
\end{figure}

\paragraph{Method Call}

The operational semantics for method invocation are presented in
Figure~\ref{fig:semantics_method_call} as two separate rules for
readability. The second rule, \rulename{Variable/Pool Declaration}
is only used from inside the first, \rulename{Method Call}. By
\rulename{Method Call}, a method call proceeds by constructing a
new stack frame $\sframe$ for the method $m$, populating it with the current \texttt{this}, the method parameter $x'$ and the \texttt{this}' pool parameters.
In a big-step way, it then proceeds to evaluate the method's body in the context
of the new stack frame
using \rulename{Variable/Pool Declaration} and returning the resulting value $\beta$. 

\rulename{Variable/Pool declaration} initialises the
method's local variables. For simplicity, we require all local
variables to be declared upfront and initiaise local object variables $x$ 
to $\nullk$. For pool variables $y$, new (empty) pools are
constructed in a two-step manner: The pools are first reserved on the
heap and then they are actually constructed, along with the stack
frame. This enable cycles among
pools.

\begin{figure}[!t]
{
\begin{gather*}
\hlrule{Method Call}{
  \mathit{getThis}(\heap, \sframe(x)) = (C, \many\pool, \hlv)
  \quad
  \mlookup C m = (\_, x': \_, \lds, e)
  \\
  \reduces{\heap,
    [\thisk\mapsto\hlv, x'\mapsto\sframe(x''), \owners C \mapsto \many\pool],
    \lds; e}{\heap', \_,\hlv'}
}{
  \reduces{\heap, \sframe, \scallc x m {x''}}{\heap', \sframe, \hlv'}
}
\\
\hlrule{Variable/Pool Declaration}{
  \lds =
  \poolsk ~    y_1 \!\colon {\ownerT {L_1} {\many y_1}}
          ~..~ y_n \!\colon {\ownerT {L_n} {\many y_n}} ~~ \localsk ~ x_1 \!\colon \_ ~..~ x_m \!\colon \_
  \\
  \pool_1, ~..~, \pool_n \notin\heap
  \quad
  \forall i, j. ~ \pool_i = \pool_j \rightarrow i = j
  \\
  \heap' =
    \heap[\pool_1\mapsto\mathit{init}({\ownerT{L_1}{\many y_1}}, \sframe'),
          ~..~,
          \pool_n\mapsto\mathit{init}({\ownerT{L_n}{\many y_n}}, \sframe')]
  \\
  \reduces{\heap', \sframe[x_1 ~..~ x_m \mapsto \nullk^m, y_1 ~..~ y_n \mapsto \pool_1 ~..~ \pool_n], e}{\heap'', \_,\hlv'}
}{
  \reduces{\heap, \sframe, \lds; e}{\heap'', \sframe, \hlv'}
}
\end{gather*}
}

\noindent\raggedright
Where:

\noindent{$\arraycolsep=1.5pt
\begin{array}{lll}
  \mathit{getThis}(\heap, \object) &\equiv
     (C, \many\pool, \object)
    & \mathit{iff}~
    \heap(\object) = (C, \many\pool, \_)
\\
  \mathit{getThis}(\heap, (\pool, n)) &\equiv
     (C, \many\pool, (\pool, n))
    & \mathit{iff}~
    \heap(\pool) = (L, \many\pool, \_)
    \ands ~ \getlayout L \getInPos 0 = C
\\
  \mathit{init}(\ownerT L {\many y}, \sframe) &\equiv
    (L, \sframe(\many y), \emptyset^n)
    & \mathit{iff}~
    n = {|\getlayout L ~ \getInPos 1|}
\end{array}
$}
\caption{Operational semantics for method call.}
\label{fig:semantics_method_call}
\end{figure}

\subsection{Type System}
\label{sec:statics}

Our type system, in addition to ensuring well-typedness of
run-time objects, ensures that objects are allocated with the
correct layout in the correct pool. This is essential to ensure
that there can be no accesses to undefined memory.

\begin{figure}[b]
\begin{gather*}
\typeRule{Value}{
  \wftype \context {\ownerT C {\many y}}
}{
  \types \context \nullk \ownerT C {\many y}
}
~
\typeRule{Variable}{
}{
  \types \context x \context(x)
}
~
\typeRule{Assignment}{
  \types \context {x, e} {\concretetype}
}{
  \types{\context} {x = e} {\concretetype}
}
\\
\typeRule{New Object}{
  \wftype \context {\ownerT C {\many y}}
}{
  \types \context {\newc \ownerT C {\many y}} \ownerT C {\many y}
}
~
\typeRule{Field Read}{
  \types \context x  \ownerT C {\many y}
}{
  \types \context {\freadc x f} \flookup C f[\owners C/ \many y]
}
\\
\typeRule{Field Write}{
  \types \context x {\ownerT C {\many y}}
  \\
  \flookup C f[\owners C/ \many y] = \concretetype
  \\
  \types \context {x'} \concretetype
}{
  \types \context {\fwritec x f {x'}} \concretetype
}
~
\typeRule{Method Call}{
  \types \context x {\ownerT C {\many y}}
  \\
  \mlookup C m = (\concretetype, \_: \concretetype', \_, \_)
  \\
  \types \context {x'} {\concretetype'[\owners C / \many y]}
}{
  \types \context {\scallc x m {x'}} {\concretetype [\owners C / \many y]}
}
\\
\typeRule{Pool Variable}{\\
}{
  \types \context y \context(y)
}
~
\typeRule{None bound}{\\
    \context \vdash [\ownerT {C} {\many y}]
}{
    \types \context \nonek {[\ownerT {C} {\many y}]}
}
~
\typeRule{Pool bound}{
    \types \context y {\ownerT {L} {\many y}}
    \\
    \getlayout L \getInPos 0 = C
}{
    \types \context y {[\ownerT {C} {\many y}]}
}
\end{gather*}
\caption{Expression type checking.}
\label{fig:typing}
\end{figure}

\paragraph{Expression Types}

The type rules are presented in Figure~\ref{fig:typing}, and have the
form $\types \context e \concretetype$.  $\context$ maps variables to
types:
\begin{align*}
\context\in\setDef{TypingContext} &\grmeq
    x\colon T, \context \grmor \epsilon \\
T\in\setDef{Type} &\grmeq \setDef{ObjectType} \cup \setDef{PoolType} \cup \setDef{PoolBound}
\end{align*}

We distinguish three kinds of types:

\Pad
\begin{compactdesc}
  \setlength{\itemindent}{-6mm}
  \setlength{\labelwidth}{0mm}
\item[Object Types] $\ownerT C {\many y}$ where $C$
  is a class and $\many y$ are pool parameters which correspond to
  the class pool parameters of $C$.

\item[Pool Types] $\ownerT L {\many y}$ where $L$ is
  a layout. If $L$ is a layout that stores objects of class $C$,
  then $\many y$ are pool parameters which correspond to the class
  pool parameters of $C$.

\item[Pool Bounds] $[\ownerT C {\many y}]$ where $C$
  is a class and $\many y$ are pool parameters corresponding to
  the class pool parameters of $C$.

\end{compactdesc}
\Pad

Pool types characterise pool values, \ie{} pools that have been allocated
dynamically (that is, when a method has been called) and are organised
according to a layout. Pool bounds characterise class pool parameters, which
have not yet been initialised and, therefore, do not have a layout. However,
when a method is called, the class pool parameters will have pool values
assigned to them.

The rationale for bounds is that a method may be invoked on an object that could
be allocated on the $\nonek$ pool or a pool that adheres to a specific layout.
Thanks to bounds, we can write code that is agnostic on the kind of pool the
object is allocated.

By \rulename{Value}, $\nullk$ can have any well-formed object type.
By \rulename{Variable}, the type of a variable $x$ is looked-up from $\context$.
By \rulename{Assignment}, assignment to a local variable must match types. While we do not model inheritance or subtyping, these extensions are possible and straightforward.
By \rulename{New Object}, we can create objects from any well-formed type. 
By \rulename{Field Read}, looking up a field $f$ from a receiver $x$ of type $T$ requires that $f$ is in $T$. Furthermore, we must translate the pool parameters names internal to $T$, used in its typing of $f$, to the arguments to which these parameters are bound where the field lookup takes place. We use the helper function $\owners{C}$ to extract the names of the formal pool parameters of the class $C$.
Both \rulename{Field Write} and \rulename{Method Call} must apply similar substitutions to translate between the internal names of the formal parameters and the arguments used at the call-site. 
Otherwise, these rules are standard. 

Roles for pool variables are straightforward. 
By \rulename{Pool Variable}, the type of a pool variable $y$ is looked-up from $\context$.
By \rulename{None Bound}, any well-formed pool type is a bound on the $\nonek$ pool.
By \rulename{Pool Bound}, the bound of a pool $y$ is derived from its pool type.

\begin{figure}[t]
\begin{gather*}
\typeRule{Bound well-formedness}{
    \getclass C \getInPos 0 = y_1\!\colon [\ownerT {C_1} {\many y_1}] ~..~ y_n\!\colon[\ownerT {C_n} {\many y_n}]
    \\
    \forall i\in[1,n]. ~ \types \context {\many y'[i]} {[\ownerT {C_i} {\many y_i}] [y_1 ~..~ y_n / \many y']}
}{
    \context \vdash {[\ownerT {C} {\many y'}]}
}
\\
\typeRule{Type well-formedness}{
    \context \vdash {[\ownerT {C} {\many y}]}
}{
    \context \vdash {\ownerT {C} {\many y}}
}
~
\typeRule{Pool well-formedness}{
    \context \vdash {[\ownerT {C} {\many y}]}
    \quad
    \getlayout L \getInPos 0 = C
}{
    \context \vdash {\ownerT {L} {\many y}}
}
\end{gather*}
\Tigthen
\caption{Type bound well-formedness.}
\label{fig:wf_type}
\Tigthen
\end{figure}

Figure~\ref{fig:wf_type} describes well-formedness of types.
They are similar to Featherweight Java
\cite{Igarashi2001}. By \rulename{Bound Well-Formedness}, a pool bound with object type $T$ is well-formed if $T$ has all
its formal pool parameters bound to arguments of the correct type, modulo a substitution from the
names of the formal parameters to the actual arguments of $T$. 
By \rulename{Type Well-Formedness}, and \rulename{Pool Well-Formedness},
well-formed object types and pool types can be derived from well-formed pool bounds.

\subsection{Type Checking Examples}

In this section we illustrate how our type rules reject programs
that violate our designs and would therefore suffer performance
penalties (and add complexity to our implementation).

\paragraph{Pool Monomorphism}

With the exception of the $\nonek$ pool, pools in \SHAPES{} are
monomorphic, \ie{} only store objects of a single type. 

\begin{lstlisting}
class Student<<ps: [Student<<ps, pp>>], pp: [Professor<<pp, ps>>] >> {
    supervisor: Professor<<pp, ps>>;
    def generate() {
        pools stuPool: StudentSplit<<stuPool, profPool>>
              profPool: ProfessorSplit<<profPool, stuPool>>
        locals stu: Student<<stuPool, profPool>>
               prof: Professor<<profPool, stuPool>> ;

        stu = new Student<<stuPool, profPool>>;    // OK
        prof = new Professor<<stuPool, profPool>>; // Error!
    }
}
...
layout StudentSplit: [Student] = ...;
layout ProfessorSplit: [Professor] = ...;
\end{lstlisting}

\noindent
The above program is rejected because of the attempt to construct
a new \lstinline!Professor! object inside a pool of students on Line 10.

\paragraph{Pool Homogeneity}

Pools in \SHAPES{} are homogeneous. This means that two objects in
a pool cannot have equi-named fields with different types. Consequently,
all objects in a pool can share the same code for dereferencing a field.

\begin{lstlisting}
class Student<<ps: [Student<<ps, pp>>], pp: [Professor<<pp>>] >> {
    supervisor: Professor<<pp>>;
    def generate() {
        pools stuPool: StudentSplit<<stuPool, profPool1>>
              profPool1: ProfessorSplit<<profPool1>>
              profPool2: ProfessorSplit<<profPool2>>
        locals stu: Student<<stuPool, profPool1>>
               prof1: Professor<<profPool1>>
               prof2: Professor<<profPool2>> ;

        stu1 = new Student<<stuPool, profPool1>>;
        stu2 = new Student<<stuPool, profPool2>>;

        stu1.supervisor = new Professor<<profPool1>>; // OK
        stu2.supervisor = new Professor<<profPool2>>; // Error!
    }
}
...
layout StudentSplit: [Student] = ...;
layout ProfessorSplit: [Professor] = ...;
\end{lstlisting}

\noindent
The above program is rejected as it attempts to assign two
students in the same pool to professors in different pools (Line
12 \& 15). If this was legal, emitting code for
\c{student.supervisor.name}, where \c{student} could refer to
either \c{stu1} or \c{stu2}, would be forced to branch on the
layout of the supervisor at run-time.

\section{Well-formedness and Type Safety}
\label{sec:meta}


The main meta-theoretic result of this paper is the type and
memory safety of field accesses, invariant of layout changes.
Going back to Figure \ref{fig:layouts}, a programmer cannot directly reference
$(f,f')_i$ in the right-most representation. To access ``$i.f$'',
we must find the $i$th $f$ in the array. To extract the object
$(f,f')_i$, we must take care not to accidentally return
$(f_i,f'_j)$, which would be an object created out of thin air.
This is a consequence of the broken object abstraction.

Therefore, it is necessary for a type safety definition of \SHAPES{}
to show that given two aliasing references to a pooled object, the values
yielded from accessing the respective fields must be always equal. This must
also take into consideration that the types of the variables/fields that store
these references may have different conceptual types (for instance, during a
method call, the pool parameters of a variable may be renamed).

In the earlier sections, we have treated the program as implicitly
given. Here too, we assume its existence and we assume its
well-formedness. A program is well-formed if all of its class and layout
declarations are well-formed. The definition of the former is quite standard and
the definition of the latter checks that the layout declaration for a given
class considers all the fields of that class and that no field appears repeated
in different clusters. Formal definitions can be found in
Appendix~\ref{sect:wf:prog}.

Because of the above aliasing requirement, we need to define well-formedness in
such a way that the requirement holds. We also want to show that the
well-formedness definition allows us to perform a few optimising
transformations, while showing that the compiled code still preserves the
layout requirements of all pools and returns appropriate values (we leave the
definition of such a compilation for future work). Two of these optimisations
are the removal of pool parameters from objects and pools and the simplification
of pool addresses from a pool and an index $(\pool, n)$ to simply an index
($n$).

Therefore, we define a pool-aware and layout-aware definition
of well-formed configurations that is stronger than a typical definition of
well-formed configurations in object-oriented languages. Well-formed
configurations in \SHAPES{} require, among other things, that all objects in a
pool have the same class (the one required in the pool's layout type), and that
all the fields of an object point to objects which have classes and are in pools
as described in the object's type. We formally define
well-formed configuration in \S\,\ref{sec:wfc}.

We state type safety for \SHAPES{} in Theorem~\ref{soundness}, in the
sense that if a well-formed configuration takes a reduction step, then
the resulting configuration is well-formed too.
\begin{theorem}[\label{soundness}Type Safety]
For a well-formed program with heap $\heap$, stack
frame $\sframe$, corresponding typing context $\context$, and
expression $e$,

\Pad
\textbf{If }
$
  \wfHL{\context}{\heap}{\sframe}
  \ands
    \types{\context}{e}{\ownerT C {\many x}}
  \ands
  \reduces{\heap,\sframe,e}{\heap',\sframe',\hlv}
$

\textbf{then }
$
  \wfHL{\context}{\heap'}{\sframe'}
  \ands
  \wfweakOK{\heap'}{\hlv}{\ownerT C {\sframe'(\many x)}}
$
\end{theorem}

\begin{proof}[Proof sketch]
By structural induction over the derivation
$\reduces{\heap,\sframe,e}{\heap',\sframe',\hlv}$.
\end{proof}

\section{Related Work}
\label{sec:related}

For an extended coverage of related work, see Franco et al. \cite{Onward2017}

\paragraph{Memory layouts}

The idea of data placement to reduce cache misses was first introduced by Calder
et al.~\cite{calder-1998}, applying profiling techniques to find
temporal relationships among objects.

This work was then followed up by Lattner et
al.~\cite{lattner-2003,lattner-2005a} where rather than relying on profiling,
static analysis of C and C++ programs finds what layout to use.
Huang et al.~\cite{huang-2004} explore pool allocation in the context of Java.

Ureche et al. \cite{Ureche2015} present a Scala extension that allows
automatic changes to the data layout. The developer defines transformations
and the compiler applies the transformation during code generation.

\paragraph{Heap partitioning}

Deterministic Parallel Java also provides means to split data in the heap:
Java code is annotated with regions
information used to calculate the effects of reading
and writing to data~\cite{BocchinoETAL2009DPJ}. Loci~\cite{loci} split the
heap into per-thread subheaps plus a shared space.
Note that these languages only divide the heap \emph{conceptually}, and do
not aim to affect representation in memory.

Jaber et al. \cite{Jaber2017} present a heap partitioning scheme that works by
inferring ownership properties between objects.

In the context of NUMA systems, Franco and Drossopoulou use annotations
to describe in which NUMA nodes the objects should be
placed~\cite{franco2015behavioural},with the aim to improve program performance by reducing
memory accesses \emph{to remote nodes}, ignoring any possible in-cache data accesses.

\paragraph{Formalisation}

The \SHAPES{}{} type system is influenced by Ownership types
\cite{clarke_potter_noble:ownership_types} but uses
pools rather than ownership contexts, and without
enforcing a hierarchical decomposition of the object graph (that
is, we allow and handle cycles between pools).

As mentioned above, the concept of bounds and well-formed types is drawn from
Featherweight Java \cite{Igarashi2001}, with the exception that our formalism
does not have any concepts of polymorphism.

Formalisms for automatic data transformations with regards to functional
languages also exist. Leroy \cite{Leroy1992} presents a formalised
transformation of ML programs that allows them to make use of unboxed
representations. Thiemann \cite{Thiemann1995} extends on this work and
Shao \cite{Shao1997} generalises it.

Petersen et al \cite{Petersen2003} describe a model that uses ordered type
theory to define how constructs in high-level languages are laid out in memory.
This allows the runtime to achieve optimisations such as the coalescing of
multiple calls to the allocator.

\section{Final Remarks}
\label{sec:final}

We have presented a formal model (operational semantics and type
system) of \SHAPES{}, an object-oriented programming language that provides
first class support for object splitting and pooling. We have also provided
the definition for a well-formed runtime configuration of \SHAPES{} and justified
our design decisions.

We include an extended discussion of future work in \cite{Onward2017}.
Our next step is showing that the well-formedness of a configuration
is preserved during execution. We will also provide a translation to a low-level
language that will be designed in such a way so as to achieve reasonable
performance and preserve the operational semantics of \SHAPES{} code.

\section{Acknowledgements}

We would like to thank Christabel Neo and the FTfJP reviewers for their
feedback, and in particular, the very useful suggestions on how to make our
explanations and notation better and easier to understand.

\bibliographystyle{abbrv} 

\appendix

\section{Lookup functions}
\label{section:lookup}
\begin{align*}
\getclass C\ &\triangleq\ (\many\pds~\many\fds~\many\mds)
~\mathit{iff}~
(\classk~\ownerT C {\many\pds} \{\many\fds~\many\mds\})\in\mathit{prog}\getInPos{0}
    \\
\owners C\ &\triangleq\ y_1~..~y_n ~\mathit{iff}~
\getclass C\getInPos{0} = (y_1\colon\_~,...,~y_n\colon\_)
    \\
\plookup C y\ &\triangleq\ \pbds ~\mathit{iff}~ (y: \pbds) \in \getclass C\getInPos{0}
    \\
 \mlookup C m\ &\triangleq\ (\concretetype, x:\concretetype', \lds, e)
~\mathit{iff}~
(\defk~m(x: \concretetype'): \concretetype~\{\lds; e\})\in\getclass C\getInPos{2}
    \\
\flookup C f\ &\triangleq\ \concretetype ~\mathit{iff}~ (f: \concretetype)\in\getclass C\getInPos{1}
    \\
\fields C\ &\triangleq\ f_1~..~f_n ~\mathit{iff}~ \getclass C\getInPos{1} = (f_1: \_~..~ f_n: \_)
    \\
\getlayout L\ &\triangleq\ (C, \many f_1~..~\many f_n)
~\mathit{iff}~ \\
&\qquad \quad (\layoutk~L: [C] = \recordk\{\many f_1\};~..~\recordk\{\many f_n\})
\in \mathit{prog}\getInPos{1}
    \\
\offset L f &\triangleq (i, j)
~\mathit{iff}~
 \getlayout L = (C, {\many{\many f}}) \ands  {\many{\many f}}[i, j] = f
    \\
\offset C f &\triangleq i ~\mathit{iff}~ \fields C [i] = f
\end{align*}

\section{Well-formed Programs}
\label{sect:wf:prog}

\begin{definition}[\label{def:wfcontext}Well-formed context]
We use the notation $\vdash \context$ to declare that a context is well-formed.
We define:
$$\vdash \context \mathit{iff} \forall (\_: T) \in \context. ~ \wftype \context T \ands
    \forall y. ~ \types y \context {[\ownerT C {\many y}]} \rightarrow \many y [0] = y$$
\end{definition}

\noindent
Given definitions~\ref{def:wfclass} and~\ref{def:wflayout}, we define
well-formed programs as follows:
\begin{definition}[Well-formed program]
A program is well-formed if all its layout and all its class
declarations are well-formed.
\begin{align*}
\vdash \mathit{prog} \mathit{iff}
  (~\forall\cds\in{\mathit prog}\getInPos 0.~
    \wfhlclass{\mathit prog}\cds~)
  \ands
  (~\forall\tds\in{\mathit prog}\getInPos 1.~
    \wfhlclass{\mathit prog}\tds~)
\end{align*}
\end{definition}

\begin{definition}[\label{def:wfclass}Well-formed class declaration]
A class $C$ is well-formed if:
\begin{compactitem}[--]
\item Their first pool parameter has to be annotated with a bound that is of the
    same class and its parameters are the same as in the class declaration
    (and in the same order). That is, if the class pool parameters of the class
    $C$ are $\owners C = y_1~..~y_n$, then $\plookup C {y_1} = [\ownerT {C}{y_1, \ldots, y_n}]$.
\item The parameter list of all pool types must only contain parameters from
    the class' pool parameter list (\ie{} $\owners C$). This means that the $\nonek$
    keyword is disallowed as a pool parameter name.
\item The fields must have class types that are well-formed against the typing
    context $\context$ where the class' formal pool parameters have their corresponding
    bounds as types. Moreover, $\context$ is well-formed.
\item All the methods have a parameter and return type that is well-formed
    against the context $\context$. Moreover, for each method, the corresponding
    method body is typeable against a context $\context'$ which is an
    augmentation of $\context$ and contains the types of $\thisk$
    variable, local pool, and object variables of the method. Moreover $\context'$
    is well-formed.
\end{compactitem}
\begin{align*}
&\wfhlclass
  {\mathit prog}
    {\classk\ \ownerT C{y_1\colon [\ownerT {C_1}{\many y_1}] ~..~ y_n\colon [\ownerT {C_n}{\many y_n}]}} ~ \{ ~ \many \fds ~ \many \mds ~ \}
\mathit{iff}
  \\&\qquad \quad
  \vdash \context \ands C_1 = C \ands \many y_1 = y_1 ~..~ y_n
  \\&\qquad \ands
  \forall f: T \in \many \fds. ~ \wftype \context T
  \\&\qquad \ands
  \forall \defk~m(x\colon\concretetype)\colon\concretetype'~\{\lds ~ ; ~ e\} \in \many \mds. [
  \\&\qquad \qquad \quad
    \wftype \context \concretetype \ands \wftype \context {\concretetype'}
  \\&\qquad \qquad \ands
    \vdash \context' \ands
    \types {\context'} e {\concretetype'}~]
    \\&\qquad \qquad \quad \mathit{where} ~
    \context' = \context, \thisk: \!\ownerT C{y_1 ~..~ y_n}, x: \!\concretetype,
    \\&\qquad \qquad \qquad \quad
        y_1' \!\colon {\ownerT {L_1} {\many y_1'}}, ~..~,  y_k' \!\colon {\ownerT {L_k} {\many y_k'}},
    \\&\qquad \qquad \qquad \quad
        x_1 \!\colon {\ownerT {C'_1} {\many y_1''}}, ~..~,  x_m \!\colon {\ownerT {C'_m} {\many y_m''}}
    \\&\qquad \qquad \qquad
        \lds = \poolsk ~ y_1' \colon {\ownerT {L_1} {\many y_1'}}
            ~..~, y_k' \colon {\ownerT {L_k} {\many y_k'}}
    \\&\qquad \qquad \qquad \qquad
            \localsk ~
            x_1 \!\colon {\ownerT {C'_1} {\many y_1''}} ~..~  x_m \!\colon {\ownerT {C'_m} {\many y_m''}}
    \\&\qquad \mathit{where} ~ \context= y_1\!\colon [\ownerT {C_1}{\many y_1}] ~..~ y_n\!\colon [\ownerT {C_n}{\many y_n}]
\end{align*}

\end{definition}

\vspace{.1in}
\noindent
We now define well-formedness of layout declarations:
\begin{definition}[Well-formed layout declaration]
\label{def:wflayout}
A layout declaration for instances of a class $C$ is well-formed iff the disjoint
union of its clusters' fields is the set of the fields declared in $C$.
\begin{align*}
& \wflayout{\mathit prog}{ \layoutk\ L\colon\pooltype C = \cls_1 ... \cls_n}
\mathit{iff}
\\ & \qquad
\lbrace \fields C \rbrace =
  \biguplus_{i\in{1~..~n}} \{ \many f~|~\cls_i = \recordk\ \{ \many f\} \}
\end{align*}
\end{definition}
This definition excludes repeated or missing fields. For example, if a class
\textit{Video} has 3 fields with names \textit{id}, \textit{likes}, \textit{views},
then both of these layout declarations are ill-formed:
\begin{lstlisting}[numbers=none]
// repeated field
layout Ill_formed_Layout1: [Video] = rec {id, likes} + rec {likes, views}
// missing field
layout Ill_formed_Layout2: [Video] = rec {id} + rec {views}
\end{lstlisting}

\section{Well-formed Configurations}
\label{sec:wfc}

As usual, an object or pool adheres to a type $\ownerT C {\many y}$ or $\ownerT
L {\many y}$, respectively, if the class of the object or pool is $C$, the pool
parameters are $\sframe(\many y)$ and all of its fields or clusters adhere to
their type. To avoid the circularities in the definition of agreement, we break
it down into weak agreement, which only ensures that the run-time type of the
object is the one specified by its class and pool parameters, and strong
agreement, which also considers the contents of the fields or clusters. In
a nutshell, strong agreement checks the vertices of the object graph, while
weak agreement checks the edges.

Although the type system uses local variables in its types, we cannot use them
in the well-formedness definition, because local pool names may change across
function calls. Thus, we use run-time types instead, where the pool parameters
are substituted with pool addresses.

\begin{definition}[Run-Time types]
A run-time type $\tau$ is a defined as follows:
\begin{align*}
    \tau \in \setDef{RuntimeType} \grmeq ~& \setDef{RuntimeClassType} \cup \setDef{RuntimeLayoutType} \\ \cup ~& \setDef{RuntimeBound} \\
    \setDef{RuntimeClassType} \grmeq ~& \ownerT C {\pool_1 ~..~ \pool_n} \\
    \setDef{RuntimePoolType} \grmeq ~& \ownerT L {\pool_1 ~..~ \pool_n} \\
    \setDef{RuntimeBound} \grmeq ~& [\ownerT {C} {\pool_1 ~..~ \pool_n}] \\
\end{align*}
\end{definition}

We now define the well-formedness of a run-time configuration:

%

\begin{definition}[Well-formed high-level configurations]

\noindent
The definition of well-formedness of a heap is as follows:

\noindent
\mymathlcl{
    \wfHeap \heap
    & \mathit{iff} & \forall \address \in \dom\heap. ~ \exists \tau. ~ \wfagrees \heap \address \tau
}

An address can be either an address to heap-allocated object, to a pool,
or to a pool-allocated object, therefore we split the definition of
well-formed address as follows:

\noindent
\mymathlcl{
    \bullet \quad \wfagrees \heap \object {\ownerT C {\many\pool}} ~ \mathit{iff} \\
    \qquad \qquad \phantom{\ands} ~ \heap(\object) = (C, {\many\pool}, \record) \ands \many\pool[0] = \nonek \\
    \qquad \qquad \ands ~ \forall i \in 1 ~..~ n. ~ \wfweakOK {\heap}{\many\pool[i - 1]}{\plookup C {y_i} [y_1 ~..~ y_n / \many\pool]} \\
    \qquad \qquad \ands ~ \forall f. ~ \wfweakOK \heap {\record[\offset C f]} {\flookup C f[\owners C / {\many\pool}]} \\
    \qquad \qquad \phantom{\ands} ~ \mathit{where} ~ y_1 ~..~ y_n = \owners C
}

\noindent
\mymathlcl{
    \bullet \quad \wfagrees \heap \pool {\ownerT L {\many\pool}} ~ \mathit{iff} \\
    \qquad \qquad \phantom{\ands} ~ \heap(\pool) = (L, {\many\pool}, \many\cluster) \ands \many\pool[0] = \pool \\
    \qquad \qquad \ands ~ \forall i \in 1 ~..~ n. ~ \wfweakOK {\heap}{\many\pool[i - 1]}{\plookup C {y_i} [y_1 ~..~ y_n / \many\pool]} \\
    \qquad \qquad \ands ~ |\cluster_1| = \ldots = |\cluster_n| \\
    \qquad \qquad \ands ~ \forall i \in 0 ~..~ |\cluster_1| - 1. ~ \wfagrees \heap {(\pool, i)}
    \ownerT {C} {\many\pool} \\
    \qquad \qquad \phantom{\ands} ~ \mathit{where} ~ y_1 ~..~ y_n = \owners C, \getlayout L \getInPos 0 = C
}

\noindent
\mymathlcl{
    \bullet \quad \wfagrees \heap {(\pool, n)} {\ownerT C {\many\pool}} ~ \mathit{iff} \\
    \qquad \qquad \phantom{\ands} ~ \heap(\pool) = (L, {\many\pool}, \many\cluster) \ands \many\pool[0] = \pool \\
    \qquad \qquad \ands ~ \getlayout L \getInPos 0 = C \\
    \qquad \qquad \ands ~ \forall f. ~ \wfweakOK \heap {\many\cluster[i, n, j]} {\flookup C f[\owners C / {\many\pool}]} \\
    \qquad \qquad \phantom{\ands} ~ \mathit{where} ~ (i, j) = \offset L f
}

\noindent
The definition of weak agreement for objects is as follows:

\noindent
\mymathlcl{
\bullet \quad \wfweakOK \heap \object {\ownerT C {\many\pool}}
    & \mathit{iff} &  \heap(\object) = (C, {\many\pool}, \_) \\
    &\ands& \many\pool[0] \neq \nonek \\
\bullet \quad \wfweakOK \heap {(\pool, n)} {\ownerT C {\many\pool}}
    & \mathit{iff} & \heap(\pool) = (L, {\many\pool}, \_) \\
    &\ands& \many\pool[0] = \pool \ands \getlayout L \getInPos 0 = C \\
\bullet \quad \wfweakOK \heap \nullk {\ownerT C \_} &&
}

\noindent
The definition of weak agreement for pools and bounds is as follows:

\noindent
\mymathlcl{
\bullet \quad \wfweakOK \heap \pool {\ownerT L {\many\pool}}
    & \mathit{iff} & \heap(\many\pool[0]) = (L, {\many\pool}, \_) \\
\bullet \quad \wfweakOK \heap \nonek {[\ownerT {C} \_]} && \\
\bullet \quad \wfweakOK \heap \pool {[\ownerT {C} {\many \pool}]}
    & \mathit{iff} & \heap(\pool) = (L, {\many\pool}, \_) \\
    &\ands& \getlayout L \getInPos 0 = C \ands \pool = \many\pool[0]
}


\noindent
Finally, the definition of well-formedness of a stack frame and a heap against a context is as follows:

\noindent
\mymathlcl{
    \wfHL \context \heap \sframe
& \mathit{iff} &
    \wfHeap \heap \\
& \land & \dom\sframe = \dom\context \\
& \land & \forall x \in \dom\sframe. ~ [ \\
&& \quad [\context(x) = {\ownerT C {\many y}} \rightarrow \wfweakOK \heap {\sframe(x)} \ownerT C {\sframe(\many y)}] ~ \land \\
&& \quad [\context(x) = {\ownerT L {\many y}} \rightarrow \wfweakOK \heap {\sframe(x)} \ownerT L {\sframe(\many y)}] ~ \land \\
&& \quad [\context(x) = {[\ownerT {C} {\many y}]} \rightarrow  \wfweakOK \heap {\sframe(x)} {[\ownerT {C} {\sframe(\many y)}]}]\\
&&]
}
\end{definition}

\end{document}